\documentclass[a4paper,UKenglish,cleveref, autoref]{lipics-v2021}
\usepackage{hyperref}

\pdfoutput=1 
\hideLIPIcs 

\title{Relational Dualities and Bisimulation} 

\author{Piotr Kozicki}{University of Bristol, United Kingdom}{piotr.kozicki@bristol.ac.uk}{https://orcid.org/0009-0003-5730-6895}{}

\author{G. A. Kavvos}{University of Bristol, United Kingdom}{alex.kavvos@bristol.ac.uk}{https://orcid.org/0000-0001-7953-7975}{}

\authorrunning{P. Kozicki and G. A. Kavvos}

\Copyright{Piotr Kozicki and G. A. Kavvos}



\relatedversion{} 

\funding{This work was supported by a grant from the Research Institute on
Verified Trustworthy Software Systems (VeTSS) and a grant from the Advanced
Research + Invention Agency (ARIA).}

\acknowledgements{The second author benefitted greatly from conversations with
Kohei Kishida.}

\nolinenumbers 

\EventEditors{Frank Pfenning}
\EventNoEds{1}
\EventLongTitle{11th International Conference on Formal Structures for Computation and Deduction (FSCD 2026)}
\EventShortTitle{FSCD 2026}
\EventAcronym{FSCD}
\EventYear{2026}
\EventDate{July 20--23, 2026}
\EventLocation{Lisbon, Portugal}
\EventLogo{}
\SeriesVolume{378}
\ArticleNo{23}

\usepackage{preamble}
\newtheorem*{notation}{Notation}

\begin{document}

\maketitle

\begin{abstract}
  The Kripke semantics of various logics arises via categorical
  dualities between a category of relational frames and their maps,
  and a category of algebras and logical homomorphisms. When the
  relational frames are considered as computational systems (e.g.\ the
  states of a machine), the corresponding algebra is one of logical
  predicates on these systems (e.g.\ predicates on these states, i.e.\
  program logics). Our aim is to extend this phenomenon to relations,
  putting well-behaved relations between systems (e.g.\ bisimulations)
  in correspondence with relations between predicates. This is
  achieved by constructing particular relational extensions of Tarski
  duality (for infinitary classical propositional logic) and Thomason
  duality (for infinitary classical modal logic). We sketch how these
  dualities give rise to a proof system that relates formulae between
  different systems.
\end{abstract}

\section{Introduction}\label{section:introduction}

It is a well-known secret that relational semantics (e.g.\ Kripke semantics)
arises from \emph{dualities}, i.e.\ categorical equivalences of the form $F :
\Op{\CC} \Equiv \DD$. The idea is that $\CC$ is a `spatial' category whose
objects are \emph{frames}, and $\DD$ is an `algebraic' category whose objects
are \emph{algebras of predicates}, which allow reasoning about the elements of
the frames.

For example, truth table semantics for infinitary classical propositional logic
arises from the \emph{Tarski duality} $\Op{\SET} \Equiv \CABA$ between the
category of sets and the category of \emph{complete atomic Boolean algebras}
(CABAs) \cite[\S 1]{kishida_2018}.\footnote{The semantics of \emph{finitary} classical
  propositional logic arises from the duality 
  \(\Op{\SET_{\text{fin}}} \Equiv \BOOL_{\text{fin}}\)
  between finite sets and finite Boolean algebras. In
  fact, we can arrive at the Tarski duality \(\Op{\SET} \Equiv \CABA\) by taking
the Pro-completion of this duality \cite{kurz_2012}.}
The Kripke semantics of infinitary classical modal logic
\cite{kripke_1963,kripke_1963b} arises from the \emph{Thomason duality}
$\Op{\FRMOPEN} \Equiv \CABAO$ between the category of Kripke frames and the category
of CABAs with operators (i.e.\ modalities) \cite{thomason_1975} \cite[\S
2.4]{kishida_2018}. The Kripke semantics of infinitary intuitionistic logic
\cite{kripke_1965b} arises from the duality $\Op{\POSET} \Equiv \PRIMEALGLATT$
between posets and prime algebraic lattices \cite{nielsen_1981,kavvos_2024a}.
Famous dualities of this form include Stone duality (for Boolean algebras)
\cite{johnstone_1982}, J\'{o}nsson-Tarski duality
\cite{jonsson_1951,hansoul_1983,sambin_1988} (for Boolean algebras with
operators), and Esakia duality \cite{esakia_2019} (for Heyting algebras). See
the chapter by Kishida \cite{kishida_2018} for further references.

The overwhelming majority of work in this area has not paid much attention to
the logical r\^{o}le of \emph{morphisms} of these categories.\footnote{One
exception is the work of Abramsky \cite{abramsky_1991} on domain theory in
logical form.} In many cases the `frame' morphisms $f : W \to W'$ of $\CC$ can
be thought of as `truth-preserving': they map worlds to worlds in a way that
preserves and reflects truth, i.e.\ $w \vDash \varphi$ iff $f(w) \vDash
\varphi$. The corresponding morphism $\Pre{f} : F(W') \to F(W)$ of $\DD$ then
preserves the logical structure, e.g.\ $\Pre{f}(\varphi \land \psi) =
\Pre{f}(\varphi) \land \Pre{f}(\psi)$. Thus $\Pre{f}$ can be considered as
mapping predicates over $W'$ to predicates over $W$, in a manner that
preserves logical connectives.

Obtaining this property in the case of intuitionistic and modal logics requires
strong properties on frame morphisms. For example, in the Kripke semantics of
modal logic the induced map $\Pre{f} : \Powerset*{W'} \to \Powerset*{W}$
preserves the modality $\Box$ iff $f : W \to W'$ is an \emph{open map}, i.e.\ a
transition-preserving function whose graph is a \emph{(functional) bisimulation}
\cite[\S 2.4]{kishida_2018}.\footnote{Open maps are also called
  \emph{p-morphisms} \cite[\S 2.3]{chagrov_1996} or \emph{bounded morphisms}
\cite[\S 2.1]{blackburn_2001}.}

However, functionality is incidental to the idea of bisimulation. It is natural
to wonder what happens if we relax the morphisms of $\CC$, requiring them to
merely be \emph{relations} instead of functions. Answering this question would
allow us to have categories $\CC$ consisting of frames and well-behaved
relations, e.g.\ bisimulations. A duality of this form would then put
bisimulations in correspondence with a different kind of morphism between
algebras, enabling reasoning `across' a bisimulation.

Such results, which may be called \emph{relational dualities}, can already be found
in the literature. However, most of them put relations between frames in
correspondence with \emph{hemimorphisms}, i.e.\ morphisms that preserve
most---but not all!---logical connectives. For example, a relation $R \subseteq
W \times W'$ uniquely corresponds to a function $W \to \Powerset*{W'}$, which uniquely
corresponds to a function $\Pre{R} : \Powerset*{W'} \to \Powerset*{W}$ that
preserves all joins (but not necessarily all meets) \cite{jonsson_1951} \cite[\S
2.3]{kishida_2018}. 

This is unsatisfactory from a computational perspective. Traditional dualities
allow us to compute the `action' of a morphism of frames $f : W \to W'$ on the
\emph{syntax} of a formula. Given a formula $\varphi$ (over~$W'$) we can define
a formula $\Pre{f}(\varphi)$ (over~$W$) by recursion. For example, in Tarski
duality this would be computed by defining $\Pre{f}(\varphi \land \psi) \defeq
\Pre{f}(\varphi) \land \Pre{f}(\psi)$ and so on. Thus, $\Pre{f}$ can be pushed
through all connectives until it reaches a propositional variable $p$, where it
returns a predicate $f^{-1}(p)$ which is true at $w$ exactly when $p$ is true at
$f(w)$. Hence, a function $f : W \to W'$ allows us to map predicates on $W'$ to
predicates on $W$ in a logic-preserving manner. As hemimorphisms do \emph{not}
preserve all logical connectives, it is impossible to adapt this definition to a
relation $R \subseteq W \times W'$.

To tackle this difficulty this paper develops a different type of duality, viz.
one that puts relations $R \subseteq W \times W'$ between frames in
correspondence with relations $\mathcal{F}(R) \subseteq F(W) \times F(W')$
between predicates. This allows us to view $\mathcal{F}(R)$ as a
\emph{relational judgment} between formulae pertaining to different systems,
with $R$ playing the r\^{o}le of a `background theory' which specifies how
propositions over the two systems are related. The ultimate goal is to develop a
formal system for synthetically reasoning about the \emph{relation between
formulae over different frames}. It is not difficult to imagine that this might
have applications in program logics.

If the relation $R \subseteq W \times W'$ is also a \emph{bisimulation} between
frames, it can be thought of as a notion of equivalence between the states of
two systems, with $w \Rel{R} w'$ meaning that $w$ and $w'$ have the same
behaviours. Viewed as a judgment, the induced relation $\mathcal{F}(R)
\subseteq F(W) \times F(W')$ then satisfies `inference rules,' which we wish to identify.

As we are interested in bisimulations, we take a hint from the relevant
literature on coalgebra \cite{thijs_1996,levy_2011,jacobs_2016} and base all the
results in this paper on the \emph{lower lifting}, which maps a relation $R
\subseteq W \times W'$ to the \emph{lower relation} \(\Lower{R} \subseteq
\Powerset*{W} \times \Powerset*{W'}\), defined by
\[
  S \Lower{R} T\ \defequiv\ \forall s \in S.\ \exists t \in T.\ s \Rel{R} t.
\]
Famously, this construction is \emph{not} a functor $\REL \fto \REL$, as it does
not preserve identities. However, we can define an appropriate category
$\CABAREL$ with CABAs as objects, so that it becomes a functor $\Lower : \REL
\fto \CABAREL$. Surprisingly, this functor is an equivalence. Moreover, $\Lower$
`extends' Tarski duality, in the sense that the diagram
\[
  \begin{tikzcd}
    \Op{\SET} \ar[r, hook] \ar[d, "\Equiv" left]
    & \Op{\REL} \ar[d, "\Equiv" left, "\Lower \circ (-)^\dagger" right] \\
    \CABA \ar[r, hook] & \CABAREL
  \end{tikzcd}
\]
commutes, where ${(-)}^\dagger : \Op{\REL} \cong \REL$ is the obvious formal
duality that reverses a relation. This follows a known pattern wherein
relational dualities are not really `dualities' per se, but become so only after
composition with a formal duality \cite[\S 2.3]{kishida_2018} \cite[Remark
4.1]{kurz_2023}. In this light the $\Op{(-)}$ is merely an artifact of
functionality.

After recalling some definitions to fix notation
(\cref{section:preliminaries}), we briefly recap the elements of the Tarski
duality $\Op{\SET} \Equiv \CABA$ (\cref{section:tarski}). Then, we show that the
lower lifting given above can be used to extend it to a \emph{relational Tarski
duality} (\cref{section:relational-tarski}). We achieve this by precisely
characterising the relations between CABAs that are obtained as lower liftings,
and hence establish an equivalence. Furthermore, we show that this is an
extension of the usual Tarski duality, and sketch the formal system that this
duality induces.

Following this, we show how the same lifting gives rise to a relational version
of the Thomason duality $\Op{\FRMOPEN} \Equiv \CABAO$
(\cref{section:relational-thomason}). This is somewhat more challenging: while
the Tarski duality acts on arbitrary functions, the Thomason duality associates
open maps (i.e.\ functional bisimulations) with modal complete Boolean
homomorphisms. We show that this has an analogue on the relational level,
putting (bi)simulations in unique correspondence with what we call
\emph{(bi)simulatory relations} between CABAs. We show that this
\emph{relational Thomason duality} extends the duality originally given by
Thomason, and again sketch the formal system that this duality induces.
We close with a discussion of related work (\cref{section:related-work}).

\section{Preliminaries}\label{section:preliminaries}

A \emph{relation} $R : A \relto B$ from a set $A$ to a set $B$ is a subset $R
\subseteq A \times B$. We write $a \Rel{R} b$ to mean $(a, b) \in R$. If $R : A
\relto B$ and $S : B \relto C$ their composition $R \then S : A \relto C$ is
defined in the usual manner, i.e.\ $a \Rel*{R \then S} c$ just if there exists
some $b \in B$ with $a \Rel{R} b \Rel{S} c$. Note that we write the composition
of relations using diagrammatic order. Given a relation $R : A \relto B$ we
define its \emph{opposite} $\Rev{R} : B \relto A$ by $b \Rel{\Rev{R}} a$ iff $a
\Rel{R} b$. Sets and relations form a category $\REL$ under composition, with
the identity $\Id_A : A \relto A$ given by $a_1 \Rel{\Id_A} a_2$ iff $a_1 =
a_2$. Moreover, the opposite construction extends to a functor \(\Rev{(-)} :
\Op{\REL} \to \REL\) which constitutes a formal duality $\Op{\REL} \cong \REL$
given by taking the opposite. This shows that $\REL$ is a \emph{dagger
category}, i.e.\ equipped with an involutive contravariant functor \cite[\S
2.3.1]{heunen_2019}.

A partial order $(P, \sqsubseteq_P)$ is a set $P$ equipped with a relation
$\mathord{\sqsubseteq_P} : P \relto P$ which is reflexive, transitive, and
anti-symmetric. We will commonly refer to partial orders by their carriers ($P$,
$Q$, \ldots) and often elide the subscript of $\sqsubseteq$. A function $f : P
\to Q$ is \emph{monotonic} just if $x \sqsubseteq_P y$ implies $f(x)
\sqsubseteq_Q f(y)$. A \emph{lattice} is a partial order that has all finite
joins and meets. We will use $\mathcal{L}, \mathcal{L}', \ldots$ to denote
lattices. A \emph{complete lattice} $\mathcal{L}$ is a partial order that has
all joins, which also implies that it has all meets. A pair of monotonic
functions $f : P \to Q$ and $g : Q \to P$ is an \emph{adjunction} (which in this
setting is also called a \emph{Galois connection}) just if $f(x) \sqsubseteq_Q
y$ iff $x \sqsubseteq_P g(y)$. This is often denoted by $f \adjoint g$. It is a
consequence that $f$ preserves any joins that exist in $P$ and $g$ preserves any
meets that exist in $Q$. The \emph{adjoint functor theorem} \cite[\S
7.34]{davey_2002} \cite[\S I.4.2]{johnstone_1982} says that if $f : P \to Q$ is
monotonic, $P$ is a complete lattice, and $f$ preserves all joins, then $f$ has
a right adjoint $g : Q \to P$. Dually, if \(g : Q \to P\) is monotonic, \(Q\) is
a complete lattice, and \(g\) preserves all meets, then \(g\) has a left adjoint
\(f : P \to Q\). We refer the reader to the book by Davey and Priestley
\cite{davey_2002} for background on orders, and to the books by Mac Lane,
Awodey, and Riehl \cite{mac_lane_1978,awodey_2010,riehl_2016} for category
theory.

\section{A Primer on Tarski Duality}\label{section:tarski}

We recap the Tarski duality $\Op{\SET} \Equiv \CABA$, and show how it induces an
infinitary classical propositional logic. Recall that a \emph{Boolean algebra}
is a distributive lattice $\mathcal{B}$ in which every element $x \in
\mathcal{B}$ has a \emph{complement} $\lnot x \in \mathcal{B}$ for which $x \lor
\lnot x = \top$ and $x \land \lnot x = \bot$ \cite{johnstone_1982}.

For any set $X$, its powerset $\Powerset*{X}$ is a Boolean algebra \cite[\S
5.2]{davey_2002}. In fact, it is a \emph{complete} Boolean algebra, i.e.\ a complete
lattice with joins and meets of $(S_i)_{i \in I}$ given by $\bigcup_{i \in I} S_i$ and
$\bigcap_{i \in I} S_i$ respectively. Moreover, $\Powerset*{X}$ is \emph{atomic}, in the
following sense.

\begin{definition}[Atom]
  Let $(P, \sqsubseteq)$ be a partial order with a bottom element $\bot$. An
  element $a \in P$ is an \emph{atom} of $P$ if $a \neq \bot$ and $x \sqsubseteq
  a$ implies that either $x = \bot$ or $x = a$. We write $\At{P}$ for the set of
  atoms of $P$.
\end{definition}

\begin{definition}[Atomic complete lattice]
  A complete lattice $\mathcal{L}$ is \emph{atomic} just if every element is
  equal to the join of atoms below it, i.e.\ for every $x \in \mathcal{L}$ we
  have $x = \bigsqcup \SetComp{ a \in \At{\mathcal{L}}}{ a \sqsubseteq x }$.
\end{definition}

A \emph{complete atomic Boolean algebra (CABA)} is then a complete Boolean
algebra $\mathcal{B}$ which is moreover atomic. For any set $X$, the powerset $\Powerset*{X}$
is a CABA: its atoms are exactly the singleton sets $\{ x \}$ for $x \in X$:
every $S \subseteq X$ can be reconstructed as $S = \bigcup_{x \in S} \{ x \}$.

A \emph{morphism of CABAs} $\Pre{f} : \mathcal{B} \to \mathcal{B}'$ is
a monotonic function that preserves all meets and joins. This
definition has some significant implications. First, as $\mathcal{B}$
and $\mathcal{B}'$ are Boolean algebras, every such
$\Pre{f} : \mathcal{B} \to \mathcal{B}'$ is a \emph{complete Boolean
  homomorphism}, i.e.\ preserves \emph{all} Boolean connectives,
including negation. Second, by the adjoint functor theorem the map
$\Pre{f} : \mathcal{B} \to \mathcal{B}'$ has both a left and a right
adjoint:
\[
  \begin{tikzcd}[column sep=huge]
    \mathcal{B}' \arrow[r, bend left=40, "f_!"] \arrow[r, bend right=40, "f_*"'] & \mathcal{B} \arrow[l, "f^*" description] \arrow[phantom, from=1-1, to=1-2, "{\rotatebox[origin=c]{-90}{$\dashv$}}" pos=0.5, shift left=3.5, description] \arrow[phantom, from=1-1, to=1-2, "{\rotatebox[origin=c]{-90}{$\dashv$}}", pos=0.5, shift right=3.5, description]
  \end{tikzcd}
\]
CABAs and their morphisms form a category $\CABA$.

Given any function $f : X \to Y$ we can define a complete Boolean homomorphism
$\Pre{f} : \Powerset*{Y} \to \Powerset*{X}$ by $\Pre{f}(B) = \SetComp{ x \in X
}{ f(x) \in B}$ \cite{lawvere_1970,mellies_2016,kishida_2018}. By the adjoint
functor theorem, $\Pre{f}$ has both a left and a right adjoint, respectively
given by
\begin{alignat*}{2}
  \LKan{f}(A) &= \SetComp{ y \in Y }{\exists x \in X.\ f(x) = y \land x \in A }
  &&= \SetComp{ f(x) }{ x \in A } \\
  \RKan{f}(A) &= \SetComp{ y \in Y }{ \forall x \in X.\ f(x) = y \Rightarrow x \in A }
  &&= \SetComp { y \in Y }{ \Pre{f}(\{ y \}) \subseteq A }.
\end{alignat*}
Thus, we obtain a functor $\Powerset : \Op{\SET} \fto \CABA$.

In order to invert this functor we use some properties of its adjoints. The
following lemma is standard: it follows from the facts that atoms and primes
coincide in CABAs, and that the left adjoint of a complete lattice
homomorphism preserves primes \cite[Lemma 1.23 and Exercise 1.3.10.e]{gehrke_2024}.

\begin{lemma}
  \label{lemma:left-adj-atoms}
  Let $\Pre{f} : \mathcal{B} \to \mathcal{B}'$ be a morphism of CABAs. Then its
  left adjoint $\LKan{f} : \mathcal{B}' \to \mathcal{B}$ maps atoms to atoms.
\end{lemma}

We can thus define a functor $\At : \CABA \to \Op{\SET}$, which maps a CABA
$\mathcal{B}$ to its set of atoms $\At{\mathcal{B}}$, and use
\Cref{lemma:left-adj-atoms} to map a morphism
$\Pre{f} : \mathcal{B} \to \mathcal{B}'$ to the restriction
$\LKan{f}|_{\At{\mathcal{B}'}} : \At{\mathcal{B}'} \to \At{\mathcal{B}}$ of its
left adjoint to atoms.

This is a pseudo-inverse to $\Powerset$, and we obtain a categorical
equivalence.

\begin{theorem}[Tarski duality]
  $\Op{\SET} \Equiv \CABA.$
\end{theorem}

The gist is that every CABA $\mathcal{B}$ is isomorphic to the powerset
$\Powerset*{\At{\mathcal{B}}}$ of its atoms: its elements are uniquely
determined by the atoms below them.

This duality induces (infinitary) classical propositional logic. Let $W$ be a
set of \emph{worlds}. This set corresponds to the powerset $\Powerset*{W}$,
which can be seen as a set of \emph{predicates}. A world $w \in W$ satisfies the
predicate $\varphi \in \Powerset*{W}$ just if $w \in \varphi$. As
$\Powerset*{W}$ is a Boolean algebra, predicates are closed under all Boolean
operations.\footnote{The fact that $\Powerset*{W}$ is a \emph{complete} Boolean
  algebra means that predicates are also closed under \emph{infinite}
conjunction $\bigwedge_{i \in I} \varphi_i$ and disjunction $\bigvee_{i \in I}
\varphi_i$.} Individual worlds $w \in W$ correspond to singleton predicates $\{
w \} \in \Powerset*{W}$ which uniquely characterise them. If we interpret $W$ as
the \emph{set of states} of a computer, logics of this ilk have been in
continuous employment since the pioneering work of Dijkstra
\cite{dijkstra_1976}.

\section{A Relational Tarski Duality}\label{section:relational-tarski}

If instead of a function $f : W \to W'$ we were to have an arbitrary
\emph{relation} $R : W \relto W'$, what would the corresponding morphism between
$\Powerset*{W}$ and $\Powerset*{W'}$ be?

It is possible to obtain such a duality by enlarging $\CABA$ to $\CABAJoin$,
whose morphisms only preserve joins, and then show that $\REL \Equiv \CABAJoin$
\cite{jonsson_1951,scott_1980} \cite[\S 2.3]{kishida_2018}. Indeed, if $h :
\Powerset*{X} \to \Powerset*{Y}$ preserves all joins, then---by writing any set
as a union of singletons---we see that $h$ is completely determined by a
function $X \to \Powerset*{Y}$, i.e.\ a relation $X \relto Y$. Hofmann and Nora
\cite[\S 4.5]{hofmann_2015} show that this duality is a special case of a
general one. However, these more general morphisms do \emph{not} preserve the
logical structure of predicates. For example it may be that $h(\varphi \land
\psi) \neq h(\varphi) \land h(\psi)$. We have thereby lost the benefit of being
able to recursively compute $h$ on the syntax of formulae.

Instead, we would like a duality whose logical side consists of predicates and
\emph{relations}. To construct one we consider a simple way of \emph{lifting}
relations to powersets \cite{kurz_2016}.

\begin{definition}[Lower Relation]
  Given a relation \(R : X \relto Y\), the associated \emph{lower relation}
  \(\Lower{R} : \Powerset*{X} \relto \Powerset*{Y}\) is defined by
  \[
    S \Lower{R} T\ \defequiv\ \forall s \in S.\ \exists t \in T.\ s \Rel{R} t.
  \]
\end{definition}

This lifting has a long history. It is one half of the \emph{Egli-Milner
lifting}, which first appeared in the powerdomain literature
\cite{plotkin_1976,plotkin_1979}. It then resurfaced in the coalgebra literature
\cite{thijs_1996}, where post-fixed-points amount to \emph{simulations} between
transition systems \cite[\S 3]{jacobs_2016}. Intuitively, \(S \Lower{R} T\)
means every state of \(S\) can be `$R$-simulated' by a state of \(T\).

\(\Lower\) maps a relation between sets to a relation between CABAs. However, it
is \emph{not} a functor $\REL \fto \REL$, as it maps the identity relation \(\id :
X \relto X\) to the containment relation \(\mathbin{\subseteq} : \Powerset*{X}
\relto \Powerset*{X}\). We will define a category \(\CABAREL\) so that
\(\Lower\) becomes a functor \(\REL \fto \CABAREL\). The relations between CABAs
will be characterised as follows.

\begin{definition}[Directionally atomic relation]
  Let $P, Q$ be partial orders, $\mathcal{L}$ be a complete lattice, and
  $\mathcal{B}, \mathcal{B}'$ be CABAs.
  \begin{enumerate}
    \item A relation \(R : P \relto Q\) is a \emph{bimodule} just when \(p'
      \sqsubseteq p \Rel{R} q \sqsubseteq q'\) implies \(p' \Rel{R} q'\).
    \item A relation \(R : \mathcal{L} \relto Q\) is \emph{left-disjunctive}
      just if \(a_i \Rel{R} b\) for all $i \in I$ implies \(\Par{\bigsqcup_{i
      \in I} a_i} \Rel{R} b\).
    \item A relation \(R : \mathcal{B} \relto \mathcal{B}'\) is
      \emph{atomic-founded} when for any atom \(a \in \mathcal{B}\), if \(a
      \Rel{R} b\) then there exists an atom \(b' \sqsubseteq b\) with \(a
      \Rel{R} b'\).
    \item A relation \(R : \mathcal{B} \relto \mathcal{B}'\) is
      \emph{directionally atomic} if it is left-disjunctive,
      atomic-founded, and a bimodule.
  \end{enumerate}
\end{definition}

It is simple to show that directionally atomic relations compose, and that
$\sqsubseteq$ is the identity.
\begin{proposition}
  If \(R : \mathcal{B} \relto \mathcal{B}'\) and
  \(S : \mathcal{B}' \relto \mathcal{B}''\) are directionally atomic, then so is
  their composition \(R \then S : \mathcal{B} \relto \mathcal{B}''\).
\end{proposition}

\begin{proposition}
  \(\mathbin{\sqsubseteq} : \mathcal{B} \relto \mathcal{B}\) is directionally
  atomic, and is the identity for composition of directionally atomic relations.
\end{proposition}

We can therefore define a category \(\CABAREL\) with CABAs as objects, and
directionally atomic relations as morphisms. We can then construct a functor:

\begin{lemma}\label{lemma:lower-then-cabarel}
   \(\Lower{R} : \Powerset*{X} \relto \Powerset*{Y}\) is directionally atomic
   for any relation \(R : X \relto Y\).
\end{lemma}
%
%
%

\begin{lemma}\label{lemma:lower-functor}
  \(\Lower\) is a functor \(\REL \fto \CABAREL\).
\end{lemma}
\begin{proof}

  By \Cref{lemma:lower-then-cabarel} every \(\Lower{R}\) is a morphism in
  \(\CABAREL\).

  First, we show that \(\Lower\) preserves the identity. If \(\id : X \relto X\)
  then \(\Lower{\id}\) is just the relation \(\subseteq\), which is the identity
  \(\Powerset*{X} \relto \Powerset*{X}\) in \(\CABAREL\).

  Second, we show that \(\Lower\) preserves composition. Let \(R : X \relto Y\)
  and \(S : Y \relto Z\) be relations. Suppose \(A \Lower{R \then S} C\). For
  any \(a \in A\) define
  \begin{align*}
    B_a \defeq
    \SetComp{b \in Y}{a \Rel{R} b \wedge \exists c \in C.\ b \Rel{S} c},
    &&
       B \defeq \bigcup_{a \in A} B_a.
  \end{align*}
  Then $A \Lower{R} B$ and $B \Lower{S} C$, so
  \(A \Lower{R} \then \Lower{S} C\).

  For the converse, suppose instead that \(A \Lower{R} \then \Lower{S} C\), so
  there exists some \(B\) with \(A \Lower{R} B \Lower{S} C\). Then for every
  \(a \in A\) there must exist \(b \in B\) and \(c \in C\) with
  \(a \Rel{R} b \Rel{S} c\). Hence \(A \Lower{R \then S} C\).
\end{proof}

\begin{lemma}\label{lem:lower-faithful}
  \(\Lower : \REL \fto \CABAREL\) is faithful.
\end{lemma}
\begin{proof}
  Suppose \(R, S : X \relto Y\) are two relations and that
  \(\Lower{R} = \Lower{S}\). Then \(x \Rel{R} y\) iff
  \(\Set{x} \Lower{R} \Set{y}\) iff \(\Set{x} \Lower{S} \Set{y}\) iff
  \(x \Rel{S} y\).
\end{proof}

\begin{lemma}\label{lem:cabarel-then-lower}
  \(\Lower : \REL \fto \CABAREL\) is full. In other words, any
  directionally atomic \(R : \Powerset*{X} \relto \Powerset*{Y}\) is
  \(\Lower{S}\) for some relation \(S : X \relto Y \).
\end{lemma}
\begin{proof}
  We will prove that $R = \Lower{S_R}$ where $S_R : X \relto Y$ is given by \(x
  \Rel{S_R} y \defequiv \{ x \} \Rel{R} \{ y \}\).

  Suppose that \(A \Rel{R} B\). Then for any \(a \in A\) we must have that
  \(\Set{a} \Rel{R} B\), as \(R\) is a bimodule. But then we must also have that
  there exists some \(b \in B\) such that \(\Set{a} \Rel{R} \Set{b}\), because
  \(R\) is atomic-founded, i.e.\ $a \Rel{S_R} b$. Therefore \(A \Lower{S_R} B\).

  Suppose now that \(A \Lower{S_R} B\), which is to say that for all $a \in A$ there
  exists a $b \in B$ such that $a \Rel{S_R} b$, which is to say $\{a\} \Rel{R}
  \{b\}$. Because \(R\) is a bimodule this means that for all \(a \in A\),
  \(\Set{a} \Rel{R} B\). But $A = \bigcup_{a \in A} \{ a \}$, so \(A \Rel{R}
  B\) as \(R\) is left-disjunctive.
\end{proof}

One can identify a pseudo-inverse for \(\Lower\), viz.\ the functor \(\At :
\CABAREL \fto \REL\) which restricts every relation \(R : \mathcal{L} \relto
\mathcal{L}'\) to atoms:
\begin{align*}
  \At{\mathcal{L}} &\defeq \SetComp{x \in \mathcal{L}}{x \text{ is an atom of }
                     \mathcal{L}},
  &
    \Rel{\At{R}} &\defeq \Rel{R}|_{\At{\mathcal{L}} \times \At{\mathcal{L}'}}.
\end{align*}
This is a functor: when \(a\) and \(c\) are atoms, any witness for
\(a \Relcomp{R \then S} c\) can be replaced by an atomic witness using
atomic-foundedness of \(R\) and bimodularity of \(S\).

\begin{lemma}
  \label{lemma:lower-dense}
  \(\Lower : \REL \fto \CABAREL\) is essentially surjective.
\end{lemma}
\begin{proof}
  Given a CABA $\mathcal{L}$, define relations
  \(R_\mathcal{L} : \mathcal{L} \relto \Powerset*{\At{\mathcal{L}}}\) and
  \(R^{-1}_\mathcal{L} : \Powerset*{\At{\mathcal{L}}} \relto \mathcal{L}\) by
  \(x \Rel{R_\mathcal{L}} X\) just if \(x \sqsubseteq \bigsqcup X\), and
  \(X \Rel{R^{-1}_\mathcal{L}} x\) just if \(\bigsqcup X \sqsubseteq x\).
  These relations are bimodules and left-disjunctive by monotonicity of joins.
  Atomic-foundedness of \(R_\mathcal{L}\) uses the fact that atoms in a CABA are
  completely join-prime: if \(a \sqsubseteq \bigsqcup X\), then
  \(a \sqsubseteq a'\) for some \(a' \in X\), hence \(a = a' \in X\), so
  \(\{a\} \sqsubseteq X\) witnesses the property. Atomic-foundedness of
  \(R^{-1}_\mathcal{L}\) follows by choosing the same atom.

  Their composites are the order identities. On \(\mathcal{L}\), this says
  \(x \sqsubseteq y\) iff \(x \sqsubseteq \bigsqcup X \sqsubseteq y\) for some
  set \(X\) of atoms; the backward implication is immediate, and the forward
  one follows by taking the atoms below \(y\). On
  \(\Powerset*{\At{\mathcal{L}}}\), it says \(X \subseteq Y\) iff
  \(\bigsqcup X \sqsubseteq \bigsqcup Y\), again by complete join-primality of
  atoms.
\end{proof}

\begin{theorem}
  \label{theorem:cabarel-preduality}
  \(\REL \Equiv \CABAREL\).
\end{theorem}


Of course, \Cref{theorem:cabarel-preduality} is not explicitly a duality; to
make it one we have to compose it with the formal duality \(\Rev{(-)} :
\Op{\REL} \Equiv \REL\) in the style of Kishida \cite[\S 2.3]{kishida_2018},
obtaining
\begin{theorem}[Relational Tarski duality]
  \label{theorem:rel-tarski}
  \(\Op{\REL} \Equiv \CABAREL\).
\end{theorem}
Intuitively, the fact that \(\REL\) is a dagger category means that its
morphisms have \emph{no inherent direction}. Directionality appears when we
restrict it to functions, making this equivalence a `proper' duality \cite[\S
2.3]{kishida_2018} \cite[Remark 4.1]{kurz_2023}. Indeed, it is possible to show
that this duality is an extension of the usual Tarski duality, in the sense that
there is a commutative diagram 
\begin{equation}
  \label{diagram:reltarski-extends-tarski}
  \begin{tikzcd}
    \Op{\SET} \ar[r, hook, "\Op{\Graph}"] \ar[d, "\Powerset" left, "\Equiv" right]
    & \Op{\REL} \ar[d, "\Equiv" left, "\Lower \circ (-)^\dagger" right] \\
    \CABA \ar[r, hook, swap, "j"] & \CABAREL
  \end{tikzcd}
\end{equation}
where $\Graph$ and $j$ are faithful and injective-on-objects.

We define $j$ by taking each CABA \(\mathcal{B}\) to itself, and each
complete Boolean homomorphism
\(\Pre{f} : \mathcal{B} \to \mathcal{B}'\) to
\(j(\Pre{f}) : \mathcal{B} \relto \mathcal{B}'\) given by
\(b \Rel{j(\Pre{f})} b'\) just if \(b \sqsubseteq \LKan{f}(b')\),
where $\LKan{f} \Adjoint \Pre{f}$.

\begin{lemma}
  \label{lemma:j-faithful}
  \(j : \CABA \fto \CABAREL\) is faithful and injective-on-objects.
\end{lemma}
\begin{proof}
  To begin we show that \(j(\Pre{f})\) is a morphism of \(\CABAREL\). First,
  \(j(\Pre{f})\) is left-disjunctive: suppose that for each \(i \in I\) we have
  that \(x_i \Rel{j(\Pre{f})} y\), i.e.\ \(x_i \sqsubseteq \LKan{f}(y)\). Then
  \(\bigsqcup_{i \in I} x_i \sqsubseteq \LKan{f}(y)\), and hence \(\bigsqcup_{i
  \in I} x_i \Rel{j(\Pre{f})} y\).
  Second, \(j(\Pre{f})\) is atomic-founded. Suppose \(x \Rel{j(\Pre{f})} y\),
  i.e.\ \(x \sqsubseteq \LKan{f}(y)\), with \(x\) an atom. Write $y =
  \bigsqcup_{i \in I} a_i$ for $a_i$ atoms. Then $\LKan{f}(y) = \bigsqcup_{i \in
  I} \LKan{f}(a_i)$, because $\LKan{f}$ is a left adjoint and preserves joins.
  As \(x\) is an atom, $x \sqsubseteq \LKan{f}(a_i)$ for a particular $i \in I$, so $x
  \Rel{j(\Pre{f})} a_i$ for some $a_i \sqsubseteq y$.
  Third, $j(\Pre{f})$ is a bimodule, by monotonicity of $\LKan{f}$.

  $j$ preserves identities: we have \(x \Rel{j(\id)} y\) iff \(x \sqsubseteq
  \LKan{\id}(y) = y\), viz.\ the identity in \(\CABAREL\). $j$ preserves
  composition: let \(\Pre{f} : X \to Y\) and \(\Pre{g} : Y \to Z\) be complete
  Boolean homomorphisms. Then
  \(x \Rel{j(\Pre{g} \circ \Pre{f})} z\)
  iff
  \(x \sqsubseteq \LKan{f}(\LKan{g}(z))\)
  iff
  \(x \Rel{j(\Pre{f})} \LKan{g}(z)\).
  But by reflexivity
  $\LKan{g}(z) \sqsubseteq \LKan{g}(z)$,
  i.e.
  \(\LKan{g}(z) \Rel{j(\Pre{g})} z\),
  so
  \(x \Rel*{j(\Pre{f}) \then j(\Pre{g})} z\).
  For the converse, suppose
  \(x \Rel{j(\Pre{f})} y \Rel{j(\Pre{g})} z\),
  i.e.
  \(x \sqsubseteq \LKan{f}(y)\)
  and
  \(y \sqsubseteq \LKan{g}(z)\).
  By the monotonicity of \(\LKan{f}\) we have
  \(x \sqsubseteq \LKan{f}(y) \sqsubseteq \LKan{f}(\LKan{g}(z))\),
  therefore
  \(x \Rel{j({\Pre{g} \circ \Pre{f}})} z\).

  Finally, \(j\) is faithful because \(j(\Pre{f})\) determines \(\LKan{f}\) by a
  Yoneda-type argument, and hence determines \(\Pre{f}\) too. It is
  injective-on-objects because it is the identity on objects.
\end{proof}

We define $\Graph$ by taking each set $X$ to itself, and each function
\(f : X \to Y\) to its \emph{graph} \(\Graph{f} : X \relto Y\), viz.
\(x \Graph{f} y \defequiv (f(x) = y)\).

\begin{lemma}
  \(\Graph : \SET \fto \REL\) is faithful and injective-on-objects.
\end{lemma}
\begin{proof}
  \(\Graph\) clearly preserves identity and composition. It is faithful because
  functions are entirely determined by their graph.
\end{proof}

Finally, \eqref{diagram:reltarski-extends-tarski} commutes: for
\(f : X \to Y\), both routes send \(f\), seen as a morphism
\(Y \to X\) in \(\Op{\SET}\), to the relation
\(A \subseteq \LKan{f}(B)\) between \(A \in \Powerset*{Y}\) and
\(B \in \Powerset*{X}\).

\subsection{Sketch of a Formal System}

Having developed this duality we can now see how it can be used for relating
formulae. Suppose we have a relation $R : X \relto Y$. Under the relational
Tarski duality this induces a relation
$\Lower{R} : \Powerset*{X} \relto \Powerset*{Y}$ between predicates. We denote
this relation by the judgment
\[
  \Simulates{\varphi}{R}{\psi}
\]
where $\varphi \in \Powerset*{X}$ is a predicate over $X$ and
$\psi \in \Powerset*{Y}$ over $Y$. Intuitively, this judgment says that
\emph{every state of $X$ that satisfies $\varphi$ is $R$-related to some state
of $Y$ that satisfies $\psi$}. The fact $\Lower{R}$ is a bimodule means that the
rule
\begin{mathpar}
  \inferrule{
    \varphi' \vdash \varphi \\
    \Simulates{\varphi}{R}{\psi} \\
    \psi \vdash \psi'
  }{
    \Simulates{\varphi'}{R}{\psi'}
  }
\end{mathpar}
is sound. This is reminiscent of the \emph{consequence rule} of Hoare logic
\cite{hoare_1969}.

The fact $\Lower{R}$ is left-disjunctive means that the rule
\begin{mathpar}
  \inferrule{
    \Simulates{\varphi_1}{R}{\psi} \\
    \Simulates{\varphi_2}{R}{\psi}
  }{
    \Simulates{(\varphi_1 \lor \varphi_2)}{R}{\psi}
  }
\end{mathpar}
is sound. This allows us to reason by cases on the left. The final
characteristic property of $\Lower{R}$, namely its atomic-foundedness, would
require introducing judgments that capture atomicity of predicates, in the
style of Abramsky \cite{abramsky_1991}.

Finally, the fact $\Lower$ is a functor can be formally expressed by the
rules
\begin{mathpar}
  \inferrule{
    \Simulates{\varphi}{R}{\psi} \\
    \Simulates{\psi}{S}{\chi}
  }{
    \Simulates{\varphi}{R \then S}{\chi}
  }
  \and
  \inferrule{
    \varphi \vdash \psi
  }{
    \Simulates{\varphi}{\id}{\psi}
  }
  \and
  \inferrule{
    \Simulates{\varphi}{\id}{\psi}
  }{
    \varphi \vdash \psi
  }
\end{mathpar}

\section{A Modal Relational Duality}\label{section:relational-thomason}

In this section we will extend the relational Tarski duality to (infinitary)
classical modal logic. The basis for this is the \emph{Thomason duality}
\cite{thomason_1975} \cite[\S 2.4]{kishida_2018}, which uniquely associates each
Kripke frame with a CABA equipped with a modal operator (CABAO).

A \emph{Kripke frame} \((X, R)\) (also called a \emph{transition system}) is a
set \(X\) equipped with a relation \(R : X \relto X\). The intuition is that $X$
is a set of \emph{worlds}, and the relation $R$ encodes \emph{transitions} from
one world to another. We will sometimes write $x_1 \Transition{R} x_2$ instead of
$x_1 \Rel{R} x_2$. A morphism of Kripke frames \(f : (X, R) \to (Y, S)\) is a
function \(f : X \to Y\) which preserves transitions, meaning \(x_1
\Transition{R} x_2\) implies \(f(x_1) \Transition{S} f(x_2)\). Kripke frames and
their morphisms form a category \(\FRM\).

The relation $R$ can also be seen as a map $\lambda R : X \to \Powerset*{X}$ by
currying. Then there is a unique join-preserving function
$\blacklozenge_R : \Powerset*{X} \to \Powerset*{X}$ that makes the following
diagram commute, where $\Yo : X \to \Powerset*{X}$ maps every element $x \in X$
to the singleton set $\{ x \}$.\footnote{This is also known as the \emph{left
    Kan extension} of $\lambda R$ along the Yoneda embedding.}
\[
  \begin{tikzcd}[column sep=huge,row sep=large]
    X \ar[r, "\Yo"] \ar[dr,"\lambda R" below left]
    & \Powerset*{X}
    \ar[d,dashed,bend right=40,"\blacklozenge_R" left,""{name=diamond,right}]
    \\
    & \Powerset*{X}
    \ar[u,dotted,bend right=40,"\Box_R" right,""{name=square,left}]
    \ar[phantom, from=diamond, to=square, "\Adjoint"]
  \end{tikzcd}
\]
As $\blacklozenge_R$ preserves joins and $\Powerset*{X}$ is a complete lattice,
the adjoint functor theorem implies that it has a right adjoint $\Box_R :
\Powerset*{X} \to \Powerset*{X}$, which preserves all meets.

These maps are explicitly given by
\begin{align*}
  \blacklozenge_R(A) &= \SetComp{w \in X}{\exists v \in A.\ v \Rel{R} w}, 
                     &
  \Box_R(A) &= \SetComp{w \in X}{\forall v \in X.\ w \Rel{R} v \implies v \in A}.
\end{align*}
Having one of these three pieces of data (the relation $R$; a
join-preserving map $\blacklozenge_R$; or a meet-preserving map
$\Box_R$) uniquely determines the other two \cite{kavvos_2024a}.

It is possible to define morphisms of frames purely in terms of operators:
\begin{lemma}
  \label{lemma:frm-morphism-box}
  A function $f : X \to Y$ is a morphism of frames \(f : (X, R) \to (Y, S)\) if
  and only if \(f^* \circ \Box_S \subseteq \Box_R \circ f^*\).
\end{lemma}

\begin{notation}
  In a CABAO \(\mathcal{B}\), we will usually write the operator(s) as
  \(\blacklozenge_{\mathcal{B}} \dashv \Box_{\mathcal{B}}\), unless we
  explicitly want to draw attention to the relation \(R\) (as above),
  or it is clear from context which CABA and transition relation are
  relevant.
\end{notation}

Let $\CABAOW$ be the category whose objects
$(\mathcal{B}, \Box_\mathcal{B})$ are CABAs $\mathcal{B}$ equipped
with an \emph{operator}
$\Box_{\mathcal{B}} : \mathcal{B} \to \mathcal{B}$ that preserves all
meets,\footnote{Or, equivalently, a join-preserving operator
  $\blacklozenge_{\mathcal{B}} : \mathcal{B} \to \mathcal{B}$.} and
morphisms
$\Pre{f} : (\mathcal{B}, \Box_\mathcal{B}) \to (\mathcal{B}',
\Box_{\mathcal{B}'})$ are complete Boolean homomorphisms
$\Pre{f} : \mathcal{B} \to \mathcal{B}'$ that satisfy
$\Pre{f} \circ \Box_\mathcal{B} \sqsubseteq \Box_{\mathcal{B}'} \circ
\Pre{f}$ (with the pointwise order). This yields a modal duality

\begin{theorem}[Weak Thomason duality]
  \label{theorem:thomason}
  \(\Op\FRM \Equiv \CABAOW\).
\end{theorem}

The maps of CABAOs preserve the Boolean structure, but they only preserve
$\Box$ weakly. To strengthen this we need the notion of an \emph{open
map}.

\begin{definition}
  A morphism of frames \(f : (X, R) \to (Y, S)\) is \emph{open} iff
  $f(x) \Transition{S} y'$ implies that there exists an $x' \in X$ with
  $x \Transition{R} x'$ and $f(x') = y'$.
\end{definition}

We let $\FRMOPEN$ be the wide subcategory of $\FRM$ whose morphisms
are open.  It is easy to show that such an $f$ is open iff
\(f^* \circ \Box = \Box \circ f^*\). Then, letting $\CABAO$ be the
wide subcategory of $\CABAOW$ whose morphisms $\Pre{f}$ preserve
$\Box$ (i.e.\ $\Pre{f} \circ \Box = \Box \circ \Pre{f}$) we obtain a
refinement of \Cref{theorem:thomason}:

\begin{theorem}[Thomason duality]
  \label{theorem:open-thomason}
  $\Op\FRMOPEN \Equiv \CABAO$.
\end{theorem}

To obtain a relational version of this duality we must relax the functionality
of morphisms of Kripke frames. To achieve that, notice that the notion of an
open map is precisely a \emph{functional bisimulation} \cite[\S
3.2]{sangiorgi_2009}. Hence, we replace open maps with general
\emph{(bi)simulations}.

\begin{definition}[Simulations, Cosimulations, and Bisimulations]
  Let $(X, R)$ and $(Y, S)$ be Kripke frames, and let $Q : X \relto Y$ be a
  relation.
  \begin{enumerate}
    \item $Q$ is a \emph{simulation} whenever \(x \Rel{Q} y\) and \(x
      \Transition{R} x'\) imply that we have a \(y' \in Y\) with \(y
      \Transition{S} y'\) and \(x' \Rel{Q} y'\).
    \item $Q$ is a \emph{cosimulation} whenever $Q^\dagger$ is a simulation.
    \item $Q$ is a \emph{bisimulation} whenever it is both a simulation and a
          cosimulation.
  \end{enumerate}
\end{definition}

We illustrate the definitions of simulation and cosimulation pictorially:
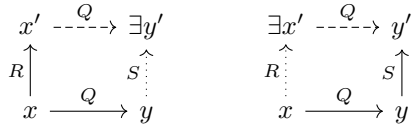
\[
  \begin{tikzcd}
    x' \ar[r,dashed,"Q"] & {\exists y'} \\
    x \ar[u,"R"] \ar[r,"Q"] & y \ar[u,dotted,"S"]
  \end{tikzcd}
  \quad\quad\quad
  \begin{tikzcd}
    {\exists x'} \ar[r,dashed,"Q"] & y' \\
    x \ar[u,dotted,"R"] \ar[r,"Q"] & y \ar[u,"S"]
  \end{tikzcd}
\]
These two conditions are often called the \emph{forth} and the \emph{back}
conditions respectively.

It is easy to see that an open map is exactly a function whose graph is a
bisimulation. Indeed, the fact it preserves transitions means it is a simulation
(i.e.\ it satisfies the `forth' condition). The fact it is open means it is a
cosimulation (i.e.\ it satisfies the `back' condition).

(Co)simulations are closed under relational composition, and the identity
relation is a bisimulation. We thus obtain two categories \(\FRMSIM\) and
\(\FRMBISIM\) with Kripke frames as objects, and simulations and bisimulations
as morphisms respectively. Note that the opposite of a simulation is a
cosimulation, and vice versa. Therefore, there is a formal duality \(\Rev{(-)} :
\Op{\FRMBISIM} \Equiv \FRMBISIM\), meaning \(\FRMBISIM\) is also a dagger
category.

It is possible to characterise simulations purely in terms of the lower
relation.

\begin{lemma}\label{lemma:sim-black}\label{lemma:simulatory}
  For a relation \(Q : (X, R) \relto (Y, S)\) between Kripke frames,
  the following are equivalent:
  \begin{enumerate}[(i)]
  \item \(Q : (X, R) \relto (Y, S)\) is a simulation.
  \item \(A \Lower{Q} B\) implies
    \(\blacklozenge A \Lower{Q} \blacklozenge B\) for any
    \(A \subseteq X\) and \(B \subseteq Y\).
  \item \(A \Lower{Q} \Box B\) implies \(\blacklozenge A \Lower{Q} B\)
    for any \(A \subseteq X\) and \(B \subseteq Y\).
  \end{enumerate}
\end{lemma}
\begin{proof}
  To show (i) implies (ii), suppose that \(Q : (X, R) \relto (Y, S)\)
  is a simulation and that \(A \Lower{Q} B\) for some
  \(A \subseteq X\) and \(B \subseteq Y\). If
  \(x \in \blacklozenge A\), then there exists some \(a \in A\) with
  \(a \Transition{R} x\). We know that there exists some \(b \in B\)
  for which \(a \Rel{Q} b\), so by simulation there must exist some
  \(y \in Y\) such that \(b \Transition{S} y\), so
  \(y \in \blacklozenge B\), and \(x \Rel{Q} y\). So
  \(\blacklozenge A \Lower{Q} \blacklozenge B\).

  To show (ii) implies (iii), suppose the premise of (ii) and that
  \(A \Lower{Q} \Box B\) for some \(A \subseteq X\) and
  \(B \subseteq Y\).  Then
  \(\blacklozenge A \Lower{Q} \blacklozenge \Box B\) by (ii). But
  \(\blacklozenge \dashv \Box\), so
  \(\blacklozenge \Box B \subseteq B\), and hence
  \(\blacklozenge A \Lower{Q} B\) as \(\Lower{Q}\) is a bimodule.

  Finally, to show (iii) implies (i), suppose the premise of (iii) and
  that \(x \Rel{Q} y\) for some \(x \in X\) and \(y \in Y\), equally
  \(\Set{x} \Lower{Q} \Set{y}\). Note that
  \(\Set{y} \subseteq \Box \blacklozenge \Set{y}\) and \(\Lower{Q}\)
  is a bimodule, therefore
  \(\Set{x} \Lower{Q} \Box \blacklozenge \Set{y}\), then by (iii)
  \(\blacklozenge \Set{x} \Lower{Q} \blacklozenge \Set{y}\).  So for
  any \(x' \in X\) for which \(x \Transition{R} x'\), there exists
  some \(y' \in Y\) for which \(y \Transition{R} y'\) and
  \(x' \Rel{Q} y'\), i.e.\ \(Q\) is a simulation.
\end{proof}

\begin{definition}
  A \emph{simulatory relation}
  \(Q : (\mathcal{B}, \Box_\mathcal{B}) \relto (\mathcal{B}',
  \Box_{\mathcal{B}'})\) between CABAOs is a directionally atomic
  \(Q : \mathcal{B} \relto \mathcal{B}'\) for which the condition of
  \Cref{lemma:simulatory} holds, i.e.\
  \[
    A \Rel{Q} \Box_{\mathcal{B}'} B\
    \Longrightarrow\
    \blacklozenge_\mathcal{B} A \Rel{Q} B.
  \]
\end{definition}

\begin{proposition}
  \label{proposition:simulatory-compose}
  If
  \(R : (\mathcal{B}, \Box_\mathcal{B}) \relto (\mathcal{B}',
  \Box_{\mathcal{B}'})\) and
  \(S : (\mathcal{B}', \Box_{\mathcal{B}'}) \relto (\mathcal{B}'',
  \Box_{\mathcal{B}''})\) are simulatory relations, then their relational
  composition is a simulatory relation
  \(R \then S : (\mathcal{B}, \Box_\mathcal{B}) \relto (\mathcal{B}'',
  \Box_{\mathcal{B}''})\).
\end{proposition}

Define \(\CABAOSIM\) to have CABAs with operators \((\mathcal{B},
\Box_\mathcal{B})\) as objects and simulatory relations
\(Q : (\mathcal{B}, \Box_\mathcal{B}) \relto (\mathcal{B}',
\Box_{\mathcal{B}'})\) as morphisms. We will often write these
morphisms simply as \(Q : \mathcal{B} \relto \mathcal{B}'\) and omit
subscripts. Using the preceding proposition, this is a category, with
\(\sqsubseteq\) as an identity morphism.

\begin{lemma}
  \label{lemma:lower-frmsim-functor}
  \(\Lower\) is a functor \(\FRMSIM \fto \CABAOSIM\).
\end{lemma}

\begin{lemma}
  \label{lemma:lower-frmsim-faithful}
  \(\Lower : \FRMSIM \fto \CABAOSIM\) is faithful.
\end{lemma}

\begin{lemma}
  \label{lemma:lower-frmsim-full}
  \(\Lower : \FRMSIM \fto \CABAOSIM\) is full. In other words, every simulatory
  relation
  \(Q : (\Powerset*{X}, \Box_R) \relto (\Powerset*{Y}, \Box_S)\)
  is \(\Lower{T}\) for some simulation \(T : (X, R) \relto (Y, S)\).
\end{lemma}

We can define a functor \(\At : \CABAOSIM \fto \FRMSIM\) by
restricting a simulatory relation to atoms, as before. Using
\Cref{lemma:simulatory} it is easy to check that \(\At{R}\) is a
simulation.

\begin{lemma}
  \label{lemma:lower-frmsim-dense}
  \(\Lower : \FRMSIM \fto \CABAOSIM\) is essentially surjective.
\end{lemma}
\begin{proof}
  Following the proof of \Cref{lemma:lower-dense},
  given a CABAO \(\mathcal{B}\), define the relations
  \(R_{\mathcal{B}} : \mathcal{B} \relto \Powerset*{\At{\mathcal{B}}}\) and
  \(R_{\mathcal{B}}^{-1} : \Powerset*{\At{\mathcal{B}}} \relto \mathcal{B}\)
  by
  \(x \Rel{R_{\mathcal{B}}} X\) iff \(x \sqsubseteq \bigsqcup X\)
  and
  \(X \Rel{R^{-1}_{\mathcal{B}}} x\) iff \(\bigsqcup X \sqsubseteq x\)
  respectively.
  As in the proof of \Cref{lemma:lower-dense} these are both directionally
  atomic and evidently inverses to each other. Thus, it suffices to show that
  they are simulatory.

  Suppose that \(x \sqsubseteq \bigsqcup \Box X\), then
  \[
    x
    \sqsubseteq \bigsqcup \Box X
    \sqsubseteq \Box \blacklozenge \bigsqcup \Box X
    = \Box \bigsqcup \blacklozenge \Box X
    \sqsubseteq \Box \bigsqcup X
  \]
  by properties of the adjunction $\blacklozenge \Adjoint \Box$.
  Hence $\blacklozenge x \sqsubseteq \bigsqcup X$, and $\Rel{R_\mathcal{B}}$ is
  simulatory.

  Similarly, if \(\bigsqcup X \sqsubseteq \Box x\), then \(\blacklozenge
  \bigsqcup X \sqsubseteq x\) and therefore \(\bigsqcup \blacklozenge X
  \sqsubseteq x\), so \(R_{\mathcal{B}}^{-1}\) is simulatory.
\end{proof}

In summary, we obtain an equivalence

\begin{theorem}
  \label{theorem:rel-thomason}
  $
  \FRMSIM \Equiv \CABAOSIM
  $.
\end{theorem}

As in our previous relational duality, the directionality disappears in the
relational case. We could turn this result into a duality by composing it with a formal
duality between $\FRMSIM$ and the category of \emph{co}simulations. However,
this is somewhat awkward.

Our final objective is to restrict this equivalence by analogy to the
restriction of \(\Op\FRM \Equiv \CABAOW\) to \(\Op\FRMOPEN \Equiv \CABAO\),
where the maps in \(\Op\FRMOPEN\) are open, and correspondingly the maps in
\(\CABAO\) preserve all logical connectives. In our relational duality the
morphisms are no longer functional, but merely simulations. Thus, openness will
turn them into \emph{bisimulations}.

Unfortunately, expressing bisimulations using our current vocabulary is not
immediately possible. For that we will need the following `dual' relational
lifting.

\begin{definition}[Upper Relation]
  Given a relation \(R : X \relto Y\), the associated upper relation \(\Upper{R}
  : \Powerset*{X} \relto \Powerset*{Y}\) is defined by
  \[
    S \Upper{R} T \defequiv \forall t \in T.\ \exists s \in S.\ s \Rel{R} t.
  \]
\end{definition}
It is easy to see that $\Upper{R} = \Lower{R^\dagger}^\dagger$ for any $R : X \relto Y$.
This makes the theory dual, in the sense that $\Upper{R}$ is
\emph{opdirectionally atomic} (i.e.\ right-disjunctive, atomic op-founded, and an
opbimodule). To avoid this proliferation of concepts we will simply work with
the opposite relation $\Lower{\Rev{R}}$. We can construct this relation directly
on relations between CABAs. 

\begin{definition}
  For a relation \(R : \mathcal{B} \relto \mathcal{B}'\) between CABAs we define
  its \emph{variant relation} \(\Variant{R} : \mathcal{B}' \relto \mathcal{B}\)
  to be
  \[
    b' \Rel{\Variant{R}} b
    \defequiv
      \forall \atom a' \sqsubseteq b'.\
      \exists \atom a \sqsubseteq b.\
      a \Rel{R} a'.
  \]
\end{definition}

This acts in the expected way between lifted relations.
\begin{lemma}
  \label{lemma:variant-lower}
  \(\Variant{\Lower{R}} = \Lower{\Rev{R}}\) for any \(R : X \relto Y\), and
  hence \(\Upper{R} = \Rev{\Variant{\Lower{R}}}\).
\end{lemma}
\begin{proof}
  Let \(R : X \relto Y\) be a relation between sets. Then for any
  \(T \subseteq Y\) and \(S \subseteq X\) we have \(T \Variant{\Lower{R}} S\)
  iff \(\forall t \in T.\ \exists s \in S.\ s \Rel{R} t\). Note that this is
  exactly the same as \(T \Lower{\Rev{R}} S\).
\end{proof}
As every directionally atomic relation is in the image of $\Lower$,
\cref{lemma:variant-lower} implies that the variant of a directionally atomic
relation is directionally atomic. Hence, the variant construction forms a
functor \(\Variant{(-)} : \Op{\CABAREL} \fto \CABAREL\) which sends every CABA
to itself, and every directionally atomic relation to its variant. Moreover,
\begin{theorem}
  \(\Variant{(-)} : \Op{\CABAREL} \fto \CABAREL\) is an equivalence.
\end{theorem}
Thus, we have a self-duality on $\CABAREL$, which makes the following
claims straightforward.

As with simulations and \Cref{lemma:sim-black}, it is possible to characterise
cosimulations purely in terms of the variant relation.

\begin{lemma}\label{lemma:cosim-black}\label{lemma:cosimulatory}
  For a relation \(Q : (X, R) \relto (Y, S)\) between Kripke frames,
  the following are equivalent:
  \begin{enumerate}[(i)]
  \item \(Q : (X, R) \relto (Y, S)\) is a cosimulation.
  \item \(B \Rel{\Variant{\Lower{Q}}} A\) implies
    \(\blacklozenge B \Rel{\Variant{\Lower{Q}}} \blacklozenge A\) for
    any \(B \subseteq Y\) and \(A \subseteq X\).
  \item \(B \Rel{\Variant{\Lower{Q}}} \Box A\) implies
    \(\blacklozenge B \Rel{\Variant{\Lower{Q}}} A\) for any
    \(B \subseteq Y\) and \(A \subseteq X\).
  \end{enumerate}
\end{lemma}
\begin{proof}
  \(Q\) is a cosimulation iff \(\Rev{Q}\) is a simulation, so just by
  \Cref{lemma:sim-black,lemma:variant-lower}.
\end{proof}

\begin{definition}
  A \emph{cosimulatory relation}
  \(Q : (\mathcal{B}, \Box_\mathcal{B}) \relto (\mathcal{B}',
  \Box_{\mathcal{B}'})\) between CABAOs is a directionally atomic
  relation \(Q : \mathcal{B} \relto \mathcal{B}'\) for which the
  condition of \cref{lemma:cosimulatory} holds, i.e.\
  \[
    B \Variant{Q} \Box_\mathcal{B} A\
    \Longrightarrow\
    \blacklozenge_{\mathcal{B}'} B \Variant{Q} A.
  \]
\end{definition}

The following proposition is analogous to
\Cref{proposition:simulatory-compose}.

\begin{proposition}
  \label{proposition:cosimulatory-compose}
  If
  \(R : (\mathcal{B}, \Box_\mathcal{B}) \relto (\mathcal{B}',
  \Box_{\mathcal{B}'})\) and
  \(S : (\mathcal{B}', \Box_{\mathcal{B}'}) \relto (\mathcal{B}'',
  \Box_{\mathcal{B}''})\) are cosimulatory relations, then their composition
  is a cosimulatory relation
  \(R \then S : (\mathcal{B}, \Box_\mathcal{B}) \relto (\mathcal{B}'',
  \Box_{\mathcal{B}''})\).
\end{proposition}

Call a relation \emph{bisimulatory} if it is both simulatory and
cosimulatory. Let \(\CABAOBISIM\) be the category whose objects are
CABAOs and whose morphisms are bisimulatory relations. As before,
\(\sqsubseteq\) is the identity.

The next four lemmas are the bisimulation analogues of
\Cref{lemma:lower-frmsim-functor,lemma:lower-frmsim-faithful,lemma:lower-frmsim-full,lemma:lower-frmsim-dense}.

\begin{lemma}
  \(\Lower\) is a functor \(\FRMBISIM \fto \CABAOBISIM\).
\end{lemma}

\begin{lemma}
  \(\Lower : \FRMBISIM \fto \CABAOBISIM\) is faithful.
\end{lemma}

\begin{lemma}\label{lem:lower-CABAOBISIM-full}
  \(\Lower : \FRMBISIM \fto \CABAOBISIM\) is full. In other words, every
  bisimulatory relation \(Q : (\Powerset*{X}, \Box_R) \relto (\Powerset*{Y},
  \Box_S)\) is \(\Lower{T}\) for some bisimulation \(T : (X, R) \relto (Y,
  S)\).
\end{lemma}

\begin{lemma}
  \label{lemma:lower-frmbisim-dense}
  \(\Lower : \FRMBISIM \fto \CABAOBISIM\) is essentially surjective.
\end{lemma}
\begin{proof}
  We augment the proof of \Cref{lemma:lower-frmsim-dense} by showing
  that both
  \(R_{\mathcal{B}} : \mathcal{B} \relto
  \Powerset*{\At{\mathcal{B}}}\) and
  \(R_{\mathcal{B}}^{-1} : \Powerset*{\At{\mathcal{B}}} \relto
  \mathcal{B}\), given by \(x \Rel{R_{\mathcal{B}}} X\) iff
  \(x \sqsubseteq \bigsqcup X\) and \(X \Rel{R^{-1}_{\mathcal{B}}} x\)
  iff \(\bigsqcup X \sqsubseteq x\) respectively, are also
  cosimulatory.

  First, \(X \Variant{R_{\mathcal{B}}} x\) iff for every atom \(\{c\} \subseteq
  X\) there exists an atom \(a \sqsubseteq x\) with \(a \Rel{R_{\mathcal{B}}}
  \{c\}\). By definition this means \(a \sqsubseteq c\), and since \(a\) and
  \(c\) are atoms, \(a = c\). Thus every element of \(X\) (all are atoms) is
  below \(x\), or equivalently \(\bigsqcup X \sqsubseteq x\), i.e.\ iff \(X
  \Rel{R^{-1}_{\mathcal{B}}} x\). Hence \(\Variant{R_{\mathcal{B}}} =
  R_{\mathcal{B}}^{-1}\).

  Similarly, \(x \Variant{R_{\mathcal{B}}^{-1}} X\) iff \(x
  \Rel{R_{\mathcal{B}}} X\). Hence \(\Variant{R_{\mathcal{B}}^{-1}} =
  R_{\mathcal{B}}\).

  A relation is cosimulatory precisely when its variant is simulatory, and
  \(R_{\mathcal{B}}\) and \(R_{\mathcal{B}}^{-1}\) are both simulatory by
  \Cref{lemma:lower-frmsim-dense}.
\end{proof}

We thus obtain the following equivalence.

\begin{theorem}
  \(\FRMBISIM \Equiv \CABAOBISIM\).
\end{theorem}

As before, this is not explicitly a duality, but we can obtain one from the
formal duality \((-)^\dagger : \Op{\FRMBISIM} \Equiv \FRMBISIM\).

\begin{theorem}[Relational Thomason duality]
  \label{theorem:rel-thomason-bisim}
  \(\Op{\FRMBISIM} \Equiv \CABAOBISIM\).
\end{theorem}

We can show that \Cref{theorem:rel-thomason-bisim} is an extension of the
Thomason duality; as with \eqref{diagram:reltarski-extends-tarski}, there is a
commutative diagram of functors
\begin{equation}
  \label{diagram:relthomason-extends-thomason-open}
  \begin{tikzcd}
    \Op{\FRMOPEN} \ar[r,hook,"\Op{\Graph}"] \ar[d,"\Powerset" left, "\Equiv" right]
    & \Op{\FRMBISIM} \ar[d,"\Lower \circ (-)^\dagger" right, "\Equiv" left] \\
    \CABAO \ar[r,hook,swap,"j"] & \CABAOBISIM
  \end{tikzcd}
\end{equation}
where $\Graph$ and $j$ are faithful and injective-on-objects.

We define \(j\) as before, by taking each CABAO
\((\mathcal{B}, \Box_{\mathcal{B}})\) to itself, and each complete
Boolean homomorphism \(\Pre{f} : \mathcal{B} \to \mathcal{B}'\) to
\(j(\Pre{f}) : \mathcal{B} \relto \mathcal{B}'\) given by
\(b \Rel{j(\Pre{f})} b'\) iff \(b \sqsubseteq f_!(b')\).

\begin{lemma}
  \(j : \CABAO \fto \CABAOBISIM\) is faithful and injective-on-objects.
\end{lemma}
\begin{proof}
  By extension of \cref{lemma:j-faithful}, we know that \(j(\Pre{f})\) is a
  directionally atomic relation, and that \(j\) preserves identities and
  composition, and that it is faithful. It remains only to show that
  \(j(\Pre{f})\) is bisimulatory.

  First, to show that \(j(\Pre{f})\) is simulatory, suppose that
  \(x \Rel{j(\Pre{f})} \Box y\), that is, \(x \sqsubseteq f_!(\Box y)\).
  Then
  \[
    x
    \sqsubseteq f_!(\Box y)
    \sqsubseteq \Box \blacklozenge f_!(\Box y)
    \sqsubseteq \Box f_!(\blacklozenge \Box y)
    \sqsubseteq \Box f_!(y)
  \]
  by the properties of the adjunction \(\blacklozenge \Adjoint \Box\). Hence
  \(\blacklozenge x \sqsubseteq f_!(y)\), so \(j(\Pre{f})\) is simulatory.

  Second, to show that \(j(\Pre{f})\) is cosimulatory, notice first that
  \(b \Rel{\Variant{j(\Pre{f})}} a\) iff, for any atom \(y \sqsubseteq b\) there
  exists some atom \(x \sqsubseteq a\) for which \(x \sqsubseteq f_!(y)\), which
  is the same as \(x = f_!(y)\) by the properties of atoms and recalling that
  \(f_!\) preserves atoms. Therefore, \(b \Rel{\Variant{j(\Pre{f})}} a\) iff
  \(b \sqsubseteq \Pre{f}(a)\). Hence, if
  \(b \sqsubseteq \Pre{f}(\Box a) = \Box \Pre{f}(a)\), then
  \(\blacklozenge b \sqsubseteq \Pre{f}(a)\) by adjunction, and \(j(\Pre{f})\) is
  cosimulatory. The functor is injective-on-objects because it is the identity
  on objects.
\end{proof}

As usual, \(\Graph\) takes each frame \((X, R)\) to itself and each open map
\(f : (X, R) \to (Y, S)\) to the graph \(\Graph{f} : X \relto Y\) defined by
\(x \Graph{f} y\) iff \(f(x) = y\).

\begin{lemma}
  The functor \(\Graph : \FRMOPEN \fto \FRMBISIM\) is faithful and
  injective-on-objects.
\end{lemma}
\begin{proof}
  \(\Graph\) clearly takes open maps to bisimulations, preserves identity and
  composition. It is faithful because any function is determined entirely by its
  graph, and injective-on-objects because it is the identity on objects.
\end{proof}

Finally, a direct calculation shows that
\eqref{diagram:relthomason-extends-thomason-open} commutes.

\subsection{Sketch of a Formal System}

Recall the formal system of the judgment $\Simulates{}{Q}{}$ for a relation $Q :
X \relto Y$ that we sketched in \Cref{section:relational-tarski}. Now assume
that \((X, R)\) and \((Y, S)\) are Kripke frames. We can extend this system to a
modal logic where the logics over $X$ and $Y$ have corresponding modalities
$\blacklozenge$ and $\Box$. If $Q$ is a simulation then by
\Cref{lemma:sim-black,lemma:simulatory} we can extend the system with the rules
\begin{mathpar}
  \inferrule{
    \Simulates{\varphi}{Q}{\psi}
  }{
    \Simulates{\blacklozenge \varphi}{Q}{\blacklozenge \psi}
  }
  \and
  \inferrule{
    \Simulates{\varphi}{Q}{\Box \psi}
  }{
    \Simulates{\blacklozenge \varphi}{Q}{\psi}
  }
\end{mathpar}
If $Q$ is a cosimulation, then by
\Cref{lemma:cosim-black,lemma:cosimulatory,lemma:variant-lower} we
obtain the rules
\begin{mathpar}
  \inferrule{
    \Simulates{\varphi}{\Rev{Q}}{\psi}
  }{
    \Simulates{\blacklozenge \varphi}{\Rev{Q}}{\blacklozenge \psi}
  }
  \and
  \inferrule{
    \Simulates{\varphi}{\Rev{Q}}{\Box \psi}
  }{
    \Simulates{\blacklozenge \varphi}{\Rev{Q}}{\psi}
  }
\end{mathpar}
which could also be obtained just by noting $\Rev{Q}$ is a
simulation.

\subsection{A Worked Example: Bisimulation Between Buffer Systems}

We will now show how to use these rules to reason over a particular
bisimulation. The two frames both model a \emph{buffer} holding a single natural
number. However, the two buffers have \emph{internal branching}: the first has
two different `ways' of holding a number (the left and right slots), while the
second one has three (slots A, B, and C). This might happen if the buffer is
provided by a distributed cluster of machines whose internal structure is not
observable.

Define $(X, \to_X)$ by letting
\(
  X = \Set{\textsf{empty}_X} \cup \SetComp{\textsf{L}(n), \textsf{R}(n)}{n \in \mathbb{N}}
\)
with 
\[
  \textsf{empty}_X \to_X \textsf{L}(n)
  \qquad
  \textsf{empty}_X \to_X \textsf{R}(n)
  \qquad
  \textsf{L}(n) \to_X \textsf{empty}_X
  \qquad
  \textsf{R}(n) \to_X \textsf{empty}_X
\]
for all \(n \in \mathbb{N}\). Similarly, define $(Y, \to_Y)$ by 
\(
  Y = \Set{\textsf{empty}_Y} \cup \SetComp{\textsf{A}(n), \textsf{B}(n), \textsf{C}(n)}{n \in \mathbb{N}}
\)
with
\[
  \textsf{empty}_Y \to_Y \textsf{A}(n)
  \qquad
  \textsf{empty}_Y \to_Y \textsf{B}(n)
  \qquad
  \textsf{empty}_Y \to_Y \textsf{C}(n)
\]
\[
  \textsf{A}(n) \to_Y \textsf{empty}_Y
  \qquad
  \textsf{B}(n) \to_Y \textsf{empty}_Y
  \qquad
  \textsf{C}(n) \to_Y \textsf{empty}_Y
\]
for all \(n \in \mathbb{N}\). There is a bisimulation \(Q : X \relto Y\) between
these frames given by
\[
  Q = \Set{(\textsf{empty}_X, \textsf{empty}_Y)} \cup \SetComp{(w, v)}{n \in \mathbb{N},\ w \in F_X(n),\ v \in F_Y(n)}
\]
where \(F_X(n) = \Set{\textsf{L}(n), \textsf{R}(n)}\) and \(F_Y(n) =
\Set{\textsf{A}(n), \textsf{B}(n), \textsf{C}(n)}\). In other words, full states
holding the same value are related, and so are empty states. This is evidently a
non-functional bisimulation.

We define the following predicates, parametrised by \(n \in \mathbb{N}\). On
\(X\) let \(\mathit{empty}_X = \Set{\textsf{empty}_X}\), \(L_n =
\Set{\textsf{L}(n)}\), and \(R_n = \Set{\textsf{R}(n)}\). On \(Y\) let
\(\mathit{empty}_Y = \Set{\textsf{empty}_Y}\).

We will use the following facts as the `background theory':
\begin{align*}
  \blacklozenge\, \mathit{empty}_X &= \textstyle\bigvee_{n} F_X(n)
  &
  \Box\, \mathit{empty}_X &= \textstyle\bigvee_n F_X(n)
  &
  \Simulates{L_n}{Q}{F_Y(n)}
  \\
  \blacklozenge\, \mathit{empty}_Y &= \textstyle\bigvee_n F_Y(n)
  &
  \Box\, \mathit{empty}_Y &= \textstyle\bigvee_n F_Y(n)
  &
  \Simulates{R_n}{Q}{F_Y(n)}
\end{align*}
The first two columns arise by direct observation. The last column holds because
every left/right state is simulated by a full state holding the same value.

We derive the fact that if the previous state in the first buffer was a left or right one holding $n$, then the current state is simulated by an empty state in the second buffer:
\[
  \Simulates{(\blacklozenge\, L_n \lor \blacklozenge\, R_n)}{Q}{\mathit{empty}_Y}
\]
The derivation proceeds in two symmetric branches, combined by disjunction:
\begin{mathpar}
  \inferrule*{
    \inferrule*{
      \inferrule*{
        \Simulates{L_n}{Q}{F_Y(n)}
        \\
        F_Y(n) \vdash \Box\, \mathit{empty}_Y
      }{
        \Simulates{L_n}{Q}{\Box\, \mathit{empty}_Y}
      }
    }{
      \Simulates{\blacklozenge L_n}{Q}{\mathit{empty}_Y}
    }
    \\
    \inferrule*{
      \inferrule*{
        \Simulates{R_n}{Q}{F_Y(n)}
        \\
        F_Y(n) \vdash \Box\, \mathit{empty}_Y
      }{
        \Simulates{R_n}{Q}{\Box\, \mathit{empty}_Y}
      }
    }{
      \Simulates{\blacklozenge R_n}{Q}{\mathit{empty}_Y}
    }
  }{
    \Simulates{(\blacklozenge L_n \lor \blacklozenge R_n)}{Q}{\mathit{empty}_Y}
  }
\end{mathpar}

\section{Related Work}\label{section:related-work}

A few relational dualities have been previously described in the literature.
Most are of the form $\Op{\CC} \Equiv \DD$ where $\CC$ is a category whose
morphisms are relations, whereas $\DD$ is a category of algebras with some form
of \emph{hemimorphism}, i.e.\ a morphism preserving most---but not all!---of
the logical structure. The earliest duality of this form is $\REL \Equiv \CABAJoin$.
This duality has been rediscovered multiple times, but is likely due to
J\'{o}nsson \cite{jonsson_1951}. Kishida \cite[\S 2.3]{kishida_2018} argues that
this extends to a `2-duality,' as both of these are (strict) 2-categories.
Halmos \cite{halmos_1956} extended it to the continuous case, i.e.\ a duality
between Stone spaces with continuous relations on the one hand, and Boolean
algebras with hemimorphisms on the other. Cignoli et al.\ \cite{cignoli_1991} do
something similar for Priestley spaces and continuous monotone relations. Hofmann
and Nora \cite{hofmann_2015} have proposed a general framework for such
dualities, obtaining the relational side as the Kleisli category of a suitable
monad. In later work they extended such dualities to metric structures and
quantale-enriched categories \cite{hofmann_2023}.

Kurz, Moshier, and Jung \cite{kurz_2023} present dualities that are
much closer to the flavour we employ in this paper. They achieve this
by working in an order-enriched setting, where relations can be
presented as both spans and cospans in a 2-categorical manner. They
use this to extend (well-behaved) dualities to dualities between
categories of relations, and even adjunctions of categories to
adjunctions between framed bicategories \cite{shulman_2008}. For
example, if their recipe is applied to the category $\POSET$ of posets
and monotone functions it lifts a bimodule $R : X \relto Y$ to the
relation $\mathbb{2}^R : [X, \mathbb{2}] \relto [Y, \mathbb{2}]$
between upper sets that is defined by letting $A \Rel{\mathbb{2}^R} B$
just if $x \in A$ and $x \Rel{R} y$ implies $y \in B$. The type of
lifting they obtain is very different: as pointed out by one of our
reviewers, $A \Rel{\mathbb{2}^R} B$ just if $R[A] \subseteq B$,
whereas $A \Lower{R} B$ just if $A \subseteq R^{-1}[B]$.

Birkmann, Urbat, and Milius \cite{birkmann_2024} present extensions of
categorical dualities through monoidal adjunctions and apply them to algebraic
language theories. As a corollary they obtain some relational dualities, e.g.\
between well-behaved relations of profinite ordered monoids and natural
morphisms of residuation algebras.

Malacaria \cite{malacaria_1995} presents the Thomason duality in a new light,
and shows how it can be used to give bisimulation an algebraic meaning. In
particular, the main theorem shows that two Kripke frames are bisimilar just if
their associated (dual) algebras have an isomorphic subalgebra.


\bibliography{main}

\begin{thebibliography}{10}

\bibitem{abramsky_1991}
Samson Abramsky.
\newblock Domain theory in logical form.
\newblock {\em Annals of Pure and Applied Logic}, 51(1):1--77, 1991.
\newblock \href {https://doi.org/10.1016/0168-0072(91)90065-T}
  {\path{doi:10.1016/0168-0072(91)90065-T}}.

\bibitem{awodey_2010}
Steve Awodey.
\newblock {\em Category Theory}.
\newblock Oxford Logic Guides. Oxford University Press, 2010.

\bibitem{birkmann_2024}
Fabian Birkmann, Henning Urbat, and Stefan Milius.
\newblock Monoidal extended stone duality.
\newblock In Naoki Kobayashi and James Worrell, editors, {\em Foundations of
  Software Science and Computation Structures}, volume 14574 of {\em Lecture
  Notes in Computer Science}, pages 144--165. Springer Nature Switzerland,
  2024.
\newblock \href {https://doi.org/10.1007/978-3-031-57228-9_8}
  {\path{doi:10.1007/978-3-031-57228-9_8}}.

\bibitem{blackburn_2001}
Patrick Blackburn, Maarten de~Rijke, and Yde Venema.
\newblock {\em Modal Logic}.
\newblock Cambridge Tracts in Theoretical Computer Science. Cambridge
  University Press, 2001.
\newblock \href {https://doi.org/10.1017/CBO9781107050884}
  {\path{doi:10.1017/CBO9781107050884}}.

\bibitem{chagrov_1996}
Alexander Chagrov and Michael Zakharyaschev.
\newblock {\em Modal Logic}.
\newblock Number~35 in Oxford Logic Guides. Oxford University Press, 1996.
\newblock \href {https://doi.org/10.1093/oso/9780198537793.001.0001}
  {\path{doi:10.1093/oso/9780198537793.001.0001}}.

\bibitem{cignoli_1991}
R.~Cignoli, S.~Lafalce, and A.~Petrovich.
\newblock Remarks on priestley duality for distributive lattices.
\newblock {\em Order}, 8(3):299--315, 1991.
\newblock \href {https://doi.org/10.1007/BF00383451}
  {\path{doi:10.1007/BF00383451}}.

\bibitem{davey_2002}
B.~A. Davey and H.~A. Priestley.
\newblock {\em Introduction to Lattices and Order}.
\newblock Cambridge University Press, 2nd edition, 2002.
\newblock \href {https://doi.org/10.1017/CBO9780511809088}
  {\path{doi:10.1017/CBO9780511809088}}.

\bibitem{dijkstra_1976}
Edsger~W. Dijkstra.
\newblock {\em A Discipline of Programming}.
\newblock Prentice-Hall, 1976.

\bibitem{esakia_2019}
Leo Esakia.
\newblock {\em Heyting Algebras: Duality Theory}, volume~50 of {\em Trends in
  Logic}.
\newblock Springer International Publishing, 2019.
\newblock \href {https://doi.org/10.1007/978-3-030-12096-2}
  {\path{doi:10.1007/978-3-030-12096-2}}.

\bibitem{gehrke_2024}
Mai Gehrke and Sam {van Gool}.
\newblock {\em Topological Duality for Distributive Lattices: Theory and
  Applications}.
\newblock Number~61 in Cambridge Tracts in Theoretical Computer Science.
  Cambridge University Press, 2024.
\newblock URL: \url{http://arxiv.org/abs/2203.03286}.

\bibitem{halmos_1956}
Paul~R. Halmos.
\newblock Algebraic logic, {I.} {Monadic} boolean algebras.
\newblock {\em Compositio Mathematica}, 12:217--249, 1954-1956.
\newblock Publisher: Kraus Reprint.
\newblock URL: \url{https://www.numdam.org/item/CM_1954-1956__12__217_0/}.

\bibitem{hansoul_1983}
G.~Hansoul.
\newblock A duality for boolean algebras with operators.
\newblock {\em Algebra Universalis}, 17(1):34--49, 1983.
\newblock \href {https://doi.org/10.1007/BF01194512}
  {\path{doi:10.1007/BF01194512}}.

\bibitem{plotkin_1979}
M.~C.~B. Hennessy and G.~D. Plotkin.
\newblock Full abstraction for a simple parallel programming language.
\newblock In Jiří Bečvář, editor, {\em Mathematical Foundations of
  Computer Science 1979}, volume~74 of {\em Lecture Notes in Computer Science},
  pages 108--120. Springer Berlin Heidelberg, 1979.
\newblock \href {https://doi.org/10.1007/3-540-09526-8_8}
  {\path{doi:10.1007/3-540-09526-8_8}}.

\bibitem{heunen_2019}
Chris Heunen and Jamie Vicary.
\newblock {\em Categories for Quantum Theory: An Introduction}.
\newblock Oxford University Press, 2019.
\newblock \href {https://doi.org/10.1093/oso/9780198739623.001.0001}
  {\path{doi:10.1093/oso/9780198739623.001.0001}}.

\bibitem{hoare_1969}
C.~A.~R. Hoare.
\newblock An axiomatic basis for computer programming.
\newblock {\em Communications of the {ACM}}, 12(10):576--580, 1969.
\newblock \href {https://doi.org/10.1145/363235.363259}
  {\path{doi:10.1145/363235.363259}}.

\bibitem{hofmann_2015}
Dirk Hofmann and Pedro Nora.
\newblock Dualities for modal algebras from the point of view of triples.
\newblock {\em Algebra universalis}, 73(3):297--320, 2015.
\newblock \href {https://doi.org/10.1007/s00012-015-0324-5}
  {\path{doi:10.1007/s00012-015-0324-5}}.

\bibitem{hofmann_2023}
Dirk Hofmann and Pedro Nora.
\newblock Duality theory for enriched priestley spaces.
\newblock {\em Journal of Pure and Applied Algebra}, 227(3):107231, 2023.
\newblock \href {https://doi.org/10.1016/j.jpaa.2022.107231}
  {\path{doi:10.1016/j.jpaa.2022.107231}}.

\bibitem{jacobs_2016}
Bart Jacobs.
\newblock {\em Introduction to Coalgebra: Towards Mathematics of States and
  Observation}.
\newblock Cambridge University Press, 2016.
\newblock \href {https://doi.org/10.1017/CBO9781316823187}
  {\path{doi:10.1017/CBO9781316823187}}.

\bibitem{johnstone_1982}
Peter~T. Johnstone.
\newblock {\em Stone Spaces}.
\newblock Number~3 in Cambridge Studies in Advanced Mathematics. Cambridge
  University Press, 1982.

\bibitem{jonsson_1951}
Bjarni Jonsson and Alfred Tarski.
\newblock Boolean algebras with operators. part i.
\newblock {\em American Journal of Mathematics}, 73(4):891, 1951.
\newblock \href {https://doi.org/10.2307/2372123} {\path{doi:10.2307/2372123}}.

\bibitem{kavvos_2024a}
G.~A. Kavvos.
\newblock {Two-Dimensional Kripke Semantics I: Presheaves}.
\newblock In Jakob Rehof, editor, {\em 9th International Conference on Formal
  Structures for Computation and Deduction (FSCD 2024)}, volume 299 of {\em
  Leibniz International Proceedings in Informatics (LIPIcs)}, pages
  14:1--14:23, Dagstuhl, Germany, 2024. Schloss Dagstuhl -- Leibniz-Zentrum
  f{\"u}r Informatik.
\newblock \href {https://doi.org/10.4230/LIPIcs.FSCD.2024.14}
  {\path{doi:10.4230/LIPIcs.FSCD.2024.14}}.

\bibitem{kishida_2018}
Kohei Kishida.
\newblock Categories and modalities.
\newblock In Elaine Landry, editor, {\em Categories for the Working
  Philosopher}. Oxford University Press, 2018.
\newblock \href {https://doi.org/10.1093/oso/9780198748991.003.0009}
  {\path{doi:10.1093/oso/9780198748991.003.0009}}.

\bibitem{kripke_1963}
Saul Kripke.
\newblock Semantical {Considerations} on {Modal} {Logic}.
\newblock {\em Acta Philosophica Fennica}, 16:83--94, 1963.
\newblock \href {https://doi.org/10.1002/malq.19630090502}
  {\path{doi:10.1002/malq.19630090502}}.

\bibitem{kripke_1963b}
Saul~A. Kripke.
\newblock Semantical {Analysis} of {Modal} {Logic} {I}. {Normal} {Modal}
  {Propositional} {Calculi}.
\newblock {\em Zeitschrift für Mathematische Logik und Grundlagen der
  Mathematik}, 9(5-6):67--96, 1963.
\newblock \href {https://doi.org/10.1002/malq.19630090502}
  {\path{doi:10.1002/malq.19630090502}}.

\bibitem{kripke_1965b}
Saul~A. Kripke.
\newblock Semantical {Analysis} of {Intuitionistic} {Logic} {I}.
\newblock In J.~N. Crossley and M.~A.~E. Dummett, editors, {\em Formal
  {Systems} and {Recursive} {Functions}}, volume~40 of {\em Studies in {Logic}
  and the {Foundations} of {Mathematics}}, pages 92--130. Elsevier, 1965.
\newblock \href {https://doi.org/10.1016/S0049-237X(08)71685-9}
  {\path{doi:10.1016/S0049-237X(08)71685-9}}.

\bibitem{kurz_2023}
Alexander Kurz, Andrew Moshier, and Achim Jung.
\newblock Stone {Duality} for {Relations}.
\newblock In Alessandra Palmigiano and Mehrnoosh Sadrzadeh, editors, {\em
  Samson {Abramsky} on {Logic} and {Structure} in {Computer} {Science} and
  {Beyond}}, pages 159--215. Springer International Publishing, Cham, 2023.
\newblock \href {https://doi.org/10.1007/978-3-031-24117-8_5}
  {\path{doi:10.1007/978-3-031-24117-8_5}}.

\bibitem{kurz_2012}
Alexander Kurz and Jiri Rosicky.
\newblock Strongly complete logics for coalgebras.
\newblock {\em Logical Methods in Computer Science}, Volume 8, Issue 3:1231,
  2012.
\newblock \href {https://doi.org/10.2168/LMCS-8(3:14)2012}
  {\path{doi:10.2168/LMCS-8(3:14)2012}}.

\bibitem{kurz_2016}
Alexander Kurz and Jiří Velebil.
\newblock Relation lifting, a survey.
\newblock {\em Journal of Logical and Algebraic Methods in Programming},
  85(4):475--499, 2016.
\newblock \href {https://doi.org/10.1016/j.jlamp.2015.08.002}
  {\path{doi:10.1016/j.jlamp.2015.08.002}}.

\bibitem{lawvere_1970}
F.~William Lawvere.
\newblock Equality in hyperdoctrines and comprehension schema as an adjoint
  functor.
\newblock In {\em Proceedings of the {AMS} Symposium on Pure Mathematics},
  volume~17, pages 1--14, 1970.
\newblock URL: \url{https://ncatlab.org/nlab/files/LawvereComprehension.pdf}.

\bibitem{levy_2011}
Paul~Blain Levy.
\newblock Similarity quotients as final coalgebras.
\newblock In Martin Hofmann, editor, {\em Foundations of Software Science and
  Computational Structures}, volume 6604 of {\em Lecture Notes in Computer
  Science}, pages 27--41. Springer Berlin Heidelberg, 2011.
\newblock \href {https://doi.org/10.1007/978-3-642-19805-2_3}
  {\path{doi:10.1007/978-3-642-19805-2_3}}.

\bibitem{mac_lane_1978}
Saunders Mac~Lane.
\newblock {\em Categories for the Working Mathematician}, volume~5 of {\em
  Graduate Texts in Mathematics}.
\newblock Springer New York, 2 edition, 1978.
\newblock \href {https://doi.org/10.1007/978-1-4757-4721-8}
  {\path{doi:10.1007/978-1-4757-4721-8}}.

\bibitem{malacaria_1995}
Pasquale Malacaria.
\newblock Studying equivalences of transition systems with algebraic tools.
\newblock {\em Theoretical Computer Science}, 139(1-2):187--205, 1995.
\newblock \href {https://doi.org/10.1016/0304-3975(94)00047-M}
  {\path{doi:10.1016/0304-3975(94)00047-M}}.

\bibitem{mellies_2016}
Paul-André Melliès and Noam Zeilberger.
\newblock A bifibrational reconstruction of lawvere's presheaf hyperdoctrine.
\newblock In {\em Proceedings of the 31st Annual {ACM}/{IEEE} Symposium on
  Logic in Computer Science}, pages 555--564. Association for Computing
  Machinery, 2016.
\newblock \href {https://doi.org/10.1145/2933575.2934525}
  {\path{doi:10.1145/2933575.2934525}}.

\bibitem{nielsen_1981}
Mogens Nielsen, Gordon Plotkin, and Glynn Winskel.
\newblock Petri nets, event structures and domains, {Part I}.
\newblock {\em Theoretical Computer Science}, 13(1):85--108, 1981.
\newblock \href {https://doi.org/10.1016/0304-3975(81)90112-2}
  {\path{doi:10.1016/0304-3975(81)90112-2}}.

\bibitem{plotkin_1976}
Gordon~D. Plotkin.
\newblock A powerdomain construction.
\newblock {\em {SIAM} Journal on Computing}, 5(3):452--487, 1976.
\newblock \href {https://doi.org/10.1137/0205035} {\path{doi:10.1137/0205035}}.

\bibitem{riehl_2016}
Emily Riehl.
\newblock {\em Category {Theory} in {Context}}.
\newblock Dover Publications, 2016.
\newblock URL: \url{http://www.math.jhu.edu/~eriehl/context.pdf}.

\bibitem{sambin_1988}
Giovanni Sambin and Virginia Vaccaro.
\newblock Topology and duality in modal logic.
\newblock {\em Annals of Pure and Applied Logic}, 37(3):249--296, 1988.
\newblock \href {https://doi.org/10.1016/0168-0072(88)90021-8}
  {\path{doi:10.1016/0168-0072(88)90021-8}}.

\bibitem{sangiorgi_2009}
Davide Sangiorgi.
\newblock On the origins of bisimulation and coinduction.
\newblock {\em {ACM} Transactions on Programming Languages and Systems},
  31(4):1--41, 2009.
\newblock \href {https://doi.org/10.1145/1516507.1516510}
  {\path{doi:10.1145/1516507.1516510}}.

\bibitem{scott_1980}
Dana~S. Scott.
\newblock Relating {Theories} of the {Lambda} {Calculus}.
\newblock In Jonathan~P. Seldin and J.~Roger Hindley, editors, {\em To {H}.
  {B}. {Curry}: {Essays} on {Combinatory} {Logic}, {Lambda} {Calculus}, and
  {Formalism}}. Academic Press, London, 1980.

\bibitem{shulman_2008}
Michael~A. Shulman.
\newblock Framed bicategories and monoidal fibrations.
\newblock {\em Theory and Applications of Categories}, 20(18):650--738, 2008.

\bibitem{thijs_1996}
Albert~Marchienus Thijs.
\newblock {\em Simulation and Fixpoint Semantics}.
\newblock PhD thesis, University of Groningen, 1996.
\newblock URL:
  \url{https://hdl.handle.net/11370/13d08025-29ff-4193-a7f2-ea5bcd20f15d}.

\bibitem{thomason_1975}
S.~K. Thomason.
\newblock Categories of frames for modal logic.
\newblock {\em The Journal of Symbolic Logic}, 40(3):439--442, 1975.
\newblock \href {https://doi.org/10.2307/2272167} {\path{doi:10.2307/2272167}}.

\end{thebibliography}

\end{document}